\documentclass[conference]{IEEEtran}

\usepackage{amsmath, amsthm, amssymb}
\usepackage{latexsym}
\usepackage{graphicx}
\usepackage{verbatim}
\usepackage{mathrsfs}
\usepackage{algorithmic}
\usepackage{float}
\usepackage[lofdepth, lotdepth]{subfig}
\usepackage{array}
\usepackage{url}

\theoremstyle{plain}
\newtheorem{theorem}{Theorem}
\newtheorem{lemma}{Lemma}
\newtheorem{proposition}{Proposition}

\theoremstyle{definition}
\newtheorem{definition}{Definition}

\newtheorem{remark}{Remark}

\newcommand{\bM}{{\boldsymbol{M}}}

\newcommand{\bP}{{\boldsymbol P}} 
\newcommand{\bQ}{{\boldsymbol Q}}

\newcommand{\bA}{{\boldsymbol A}}

\newcommand{\bG}{{\boldsymbol G}}
\newcommand{\bX}{\boldsymbol{X}}

\newcommand{\bu}{{\boldsymbol u}}
\newcommand{\bv}{{\boldsymbol v}}

\newcommand{\bc}{{\boldsymbol c}}
\newcommand{\bO}{{\boldsymbol 0}}

\newcommand{\bx}{{\boldsymbol{x}}}

\newcommand{\ft}{\mathbb{F}_2}
\newcommand{\fq}{\mathbb{F}_q}

\newcommand{\entropy}{{\sf H}}

\newcommand{\supp}{{\sf supp}}
\newcommand{\dist}{{\mathsf{d}}}

\newcommand{\rank}{{\mathsf{rank}_q}}

\newcommand{\cC}{{\mathscr{C}}}

\newcommand{\define}{\stackrel{\mbox{\tiny $\triangle$}}{=}}

\newcommand{\nin}{\noindent}

\newcommand{\et}{{\emph{et al.}}}

\newcommand{\vM}{{{\sf var}}(\boldsymbol{M})}
\newcommand{\vP}{{{\sf var}}(\boldsymbol{P})}

\newcommand{\nm}{\mathcal{N}^k({\bM})}

\begin{document}
\pagestyle{empty}

\title{Weakly Secure MDS Codes for Simple Multiple Access Networks}

\author{
   \IEEEauthorblockN{
     Son Hoang Dau\IEEEauthorrefmark{1},
     Wentu Song\IEEEauthorrefmark{2},
     Chau Yuen\IEEEauthorrefmark{3}}
   \IEEEauthorblockA{
     Singapore University of Technology and Design, Singapore\\     
		Emails: $\{$\IEEEauthorrefmark{1}{\it sonhoang\_dau},     
		\IEEEauthorrefmark{2}{\it wentu\_song},
     \IEEEauthorrefmark{3}{\it yuenchau}$\}$@sutd.edu.sg
 }
}
\maketitle

\begin{abstract}
We consider a simple multiple access network (SMAN),
where $k$ sources of unit rates transmit their data to a common sink via $n$ relays.  
Each relay is connected to the sink and 
to certain sources. 
A coding scheme (for the relays) is \emph{weakly secure} if a passive
adversary who eavesdrops on less than $k$ relay-sink links cannot
reconstruct the data from each source. 

We show that there exists a weakly secure maximum distance separable (MDS) coding scheme for the relays if and only if
every subset of $\ell$ relays must be collectively connected to at least $\ell+1$ sources, for all $0 < \ell < k$. Moreover, we prove that this condition can be verified in polynomial time in $n$ and $k$. 
Finally, given a SMAN satisfying the aforementioned condition, 
we provide another polynomial time algorithm to trim the network until it has a sparsest set of source-relay links that still supports a weakly secure MDS coding scheme.      
\end{abstract}

\section{Introduction}
\label{sec:intro}

A simple multiple access network (SMAN) is a two-hop network, 
where some $k$ independent sources transmit their data to a common sink 
via $n$ relays. We use $(n,k)$-SMAN to refer to such network.
An example of a $(6,4)$-SMAN is illustrated in 
Fig.~\ref{fig:SMAN-ex}. 
Simple multiple access networks were studied in the recent work of
Yao {\et}~\cite{YaoHoNita-Rotaru2011} (to model the problem of decentralized distribution of keys from a pool among the wireless nodes),
Halbawi {\et}~\cite{HalbawiHoYaoDuursma2013},
and Dau {\et}~\cite{DauSongDongYuenISIT2013, DauSongYuenISIT14,
DauSongYuen_JSAC2014}. 
The model of SMAN considered in~\cite{HalbawiHoYaoDuursma2013} is more general in the sense that the sources are assumed to have arbitrary rates. However, it was shown in~\cite{DauSongYuenISIT14, DauSongYuen_JSAC2014} that as far as 
the problem of constructing error-correcting codes for the relays is concerned, considering unit-rate
sources is sufficient. Interestingly, the code design problem
for SMAN was also shown in~\cite{DauSongYuenISIT14, DauSongYuen_JSAC2014} to be equivalent to the code design
problem for weakly secure cooperative data exchange~\cite{YanSprintson2013, YanSprintsonZelenko2014}.

Error correction for the general multiple access network was first investigated in the work of Dikaliotis {\et}~\cite{Dikaliotis-etal-2011}. 
The coding schemes derived in ~\cite{Dikaliotis-etal-2011}
are packetized over large fields, which are of sizes at least exponential in the number of sources. While SMAN is a special case of multiple access network~\cite{Dikaliotis-etal-2011}, the authors of~\cite{HalbawiHoYaoDuursma2013, DauSongYuenISIT14,
DauSongYuen_JSAC2014} focused more on designing error-correcting codes over small fields, whose sizes are linear in $n$ and $k$. 
Various new problems on balance and sparsity of the network were also investigated in~\cite{DauSongDongYuenISIT2013,
DauSongYuen_JSAC2014}. 
 
\begin{figure}[th]
\centering
\includegraphics[scale=0.85]{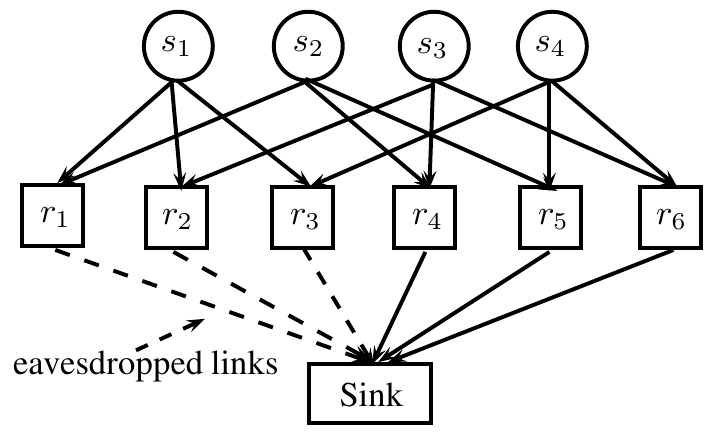}
\caption{An example of a $(6,4)$-SMAN. Three relay-sink links (dashed) 
are eavesdropped. The question is: can we prevent the adversary
from learning about each individual source packet?}
\label{fig:SMAN-ex}
\vspace{-18pt}
\end{figure}

In this paper we study the security aspect of the coding schemes
used for the relays in an $(n,k)$-SMAN. More specifically, we focus on the \emph{weak security} of such coding schemes against a \emph{passive adversary}, which eavesdrops on the relay-sink links.
Suppose that each source transmits a single packet, which is an 
element of some finite field $\fq$, to the sink. 
All source packets are assumed to be independent and randomly distributed over $\fq$. 
The coding scheme for the relays is \emph{weakly secure} if an adversary that
eavesdrops on at most $k-1$ relay-sink links gains \emph{no information}
(in Shannon's sense) about each particular source packet. 
In the context of decentralized key distribution~\cite{YaoHoNita-Rotaru2011}, a wireless node (corresponding to the sink in the SMAN) contacts its neighbors (corresponding to the relays in the SMAN) to retrieve $k$ secret keys $s_i \in \fq$ $(1 \leq i \leq k)$. 
Each of its neighbors possesses some of these $k$ keys and transmits one (coded) packet in $\fq$ to that node. In that scenario, a weakly secure coding scheme for the corresponding SMAN would guarantee that an adversary that eavesdrops on at most $k-1$ transmissions cannot determine explicitly any secret key. Note that Yao {\et}~\cite{YaoHoNita-Rotaru2011} only considered an active adversary who can 
corrupt the transmissions. In this work, we assume the presence of both an active adversary and a passive adversary. 
Note that these two adversaries may be independent of each other. 
In other words, they may attack different sets of links. 

The concept of weak security was first discussed by Yamamoto~\cite{Yamamoto1986} in the context of ramp secret sharing scheme.
After Yamamoto~\cite{Yamamoto1986}, weak security was also discovered by Bhattad and Narayanan~\cite{BhattadNarayanan2005} in a more general context of network coding. 
Weak security is important in practice since it guarantees that no meaningful information is leaked to the adversary, and often requires no additional overhead. For example, suppose that the adversary obtains the coded packet $x_1 + x_2$ where $x_1$ and $x_2$ are packets from the sources $s_1$ and $s_2$, respectively. 
Then the adversary would not be able to determine either $x_1$ or $x_2$, as from its point of view, both $x_1$ and $x_2$ are completely random variables.  

In this work we limit ourselves to maximum distance separable
(MDS) coding schemes (see Section~\ref{sec:pre} for definition). Our main contributions are summarized below. 
\begin{itemize}
	\item We establish a necessary and sufficient condition for the existence of a weakly secure MDS coding scheme for the relays. 
	More specifically, there exists
a weakly secure MDS coding scheme for the relays if and only if
every subset of $\ell$ relays must be collectively connected to at least $\ell+1$ sources, for all $0 < \ell < k$. Moreover, this condition, referred to as the \emph{Weak Security Condition}, can be verified in polynomial time
in $n$ and $k$.
\item Given a SMAN satisfying the Weak Security Condition, 
we provide a polynomial time algorithm to trim the network by removing certain source-relay links until it has the sparsest set of source-relay links that still supports a weakly secure MDS coding scheme.
This algorithm is similar to the algorithm used to find a maximum matching in a bipartite graph that deletes edges of the graph one by one until all remaining edges form a matching. 
\item We also study the so-called \emph{block security}, which is
a generalization of weak security, and characterize the block security level of an arbitrary SMAN.   
\end{itemize}
The first conclusion above describes the additional requirement on 
the source-relay links if a passive adversary is also present. 
Indeed, an MDS code implemented at the relays allows the sink
to tolerate a maximum number of $\lfloor (n-k+1)/2 \rfloor$ corrupted relay/links. Such an MDS code exists if and only if the SMAN
satisfies the MDS Condition~\cite{DauSongDongYuenISIT2013,
DauSongYuenISIT14, DauSongYuen_JSAC2014, YanSprintsonZelenko2014}: every subset of $\ell$ relays must be collectively connected to at least $\ell$ sources, for all $\ell \leq k$. Comparing the MDS Condition and the Weak Security Condition, we conclude that more 
source-relay links are required to defend both an active and
a passive adversary. Hence, a SMAN that survives the most powerful 
active adversary, which can corrupt $\lfloor (n-k+1)/2 \rfloor$ relays/links, may not be weakly secure against a passive adversary.

The paper is organized as follows. Necessary notation and definitions
are provided in Section~~\ref{sec:pre}. The weak security for SMAN
is discussed in Section~\ref{sec:weak_security}. The extension
of weak security to block security is investigated in
Section~\ref{sec:block_security}. 

\section{Preliminaries}
\label{sec:pre}

Let $\fq$ denote the finite field with $q$ elements. 
Let $[n]$ denote the set $\{1,2,\ldots,n\}$. 
For a $k \times n$ matrix $\bM$, for $i \in [k]$ and $j \in [n]$, let
$\bM_i$ and $\bM[j]$ denote the row $i$ and the column $j$ of $\bM$, respectively. 
We define below standard notions from coding theory (for instance, see \cite{MW_S}).
 
The \emph{support} of a vector $\bu = (u_1,\ldots,u_n) \in \fq^n$ is the set 
$\supp(\bu) = \{i \in [n] \colon u_i \neq 0\}$. 
The (Hamming) distance between two vectors $\bu$ and $\bv$ of $\fq^n$
is defined to be
$
\dist(\bu,\bv) = |\{i \in [n] \colon u_i \ne v_i\}|. 
$
A $k$-dimensional subspace $\cC$ of $\fq^n$ is called a linear $[n,k,d]_q$ \emph{(error-correcting) code} over $\fq$ 
if the minimum distance $\dist(\cC)$ between any pair of distinct
vectors in $\cC$ is equal to $d$.  
Sometimes we may use the notation $[n,k]_q$ or just $[n,k]$ for the sake of simplicity. The vectors in $\cC$ are called \emph{codewords}. 
A \emph{generator matrix} $\bG$ of an $[n,k]_q$ code $\cC$ is a $k \times n$ matrix whose rows are linearly independent codewords of $\cC$. Then $\cC = \{\bx \bG \colon \bx \in \fq^k\}$. 
The well-known Singleton bound (\cite[Ch. 1]{MW_S}) states that for any $[n,k,d]_q$
code, it holds that $d \leq n - k + 1$. 
If the equality is attained, the code is called \emph{maximum distance separable} (MDS).  

\vspace{-3pt}
\begin{definition} An $(n,k)$ simple multiple access network 
$((n,k)$-SMAN for short$)$ is a network that consists of
\begin{itemize}
	\item $k$ independent sources $s_1,\ldots,s_k$ of unit rates, one sink, and $n$ relays $r_1,\ldots,r_n$, where $n \geq k$, and
	\item some directed edges of capacity one that connect certain source-relay pairs and one directed
	edge of capacity one that connects each relay to the sink. 
\end{itemize}
\end{definition} \vspace{-3pt}

An $(n,k)$-SMAN can be represented by an \emph{adjacency matrix}
$\bM = (m_{i,j}) \in \ft^{k \times n}$ where $m_{i,j}=1$ if
and only if the source $s_i$ is connected to the relay $r_j$. 

Let $\bX=(X_1, \ldots, X_k)$ be a vector of independent and identically uniformly distributed random variables over $\fq$. 
We assume that the vector of source packets is $\bx = (x_1,\ldots, x_k)$, a realization of $\bX$.  
A linear coding scheme for an $(n,k)$-SMAN is represented by 
a $k \times n$ matrix $\bG = (g_{i,j})$ over $\fq$. The coding rule for the relays is as follows: $r_j$ $(j \in [n])$ 
creates and transmits the coded packet $\bx \bG[j]$ to the sink.
We refer to $\bG$ as the \emph{encoding matrix} of the coding scheme.
Note that $g_{i,j}$ must be zero whenever $m_{i,j} = 0$.  
If $\bG$ generates a linear code that can correct $t$ errors then
the sink can still determine all $k$ source packets under the presence
of at most $t$ erroneous coded packets sent from some $t$ relays.

A coding scheme based on $\bG$ is
\emph{weakly secure} if the conditional entropy \vspace{-3pt}
\[
\entropy\big(X_i \mid \{\bX \bG[j] : j \in E\}\big) = \entropy(X_i),\vspace{-3pt}
\] 
for every $i \in [k]$ and for every subset $\varnothing \neq E \subset [n]$, $|E| < k$. 
In words, a coding scheme is \emph{weakly secure} if an adversary 
that eavesdrops on at most $k-1$ coded packets transmitted on different relay-sink links obtains no information about each particular
source packet. 
Note that we always assume that $\rank(\bG) = k$. Hence, obviously an adversary that eavesdrops on certain $k$ linearly independent coded packets can always retrieve all $k$ source packets.  
 
\section{Weak Security for SMAN}
\label{sec:weak_security}

\subsection{Necessary and Sufficient Condition for Weak Security}

We first derive a necessary and sufficient condition on the links between sources and relays for a SMAN to support 
a weakly secure MDS coding scheme. 

\begin{theorem}
\label{thm:main1}
An $(n,k)$-SMAN supports a weakly secure MDS coding scheme, i.e. 
there exists a weakly secure MDS coding scheme for the relays over
some finite field $\fq$, if and only if every subset of $\ell$ relays must be collectively connected to at least $\ell+1$ sources, for all $0 < \ell < k$. In other words, it requires that
\begin{equation} 
\label{eq:weak_security}
|\cup_{j \in J} \supp(\bM[j])| \geq |J|+1,\
\forall \varnothing \neq J \subset [n],\ |J| < k,
\end{equation} 
where $\bM[j]$ is the $j$th column of the adjacency matrix $\bM$. 
We refer to (\ref{eq:weak_security}) as the Weak Security Condition for SMAN. 
\end{theorem}

We need a few lemmas for the proof of Theorem~\ref{thm:main1}. 
\begin{lemma}[\cite{VladutNoginTsfasman2007}]
\label{lem:GM_distance}
The $k\times n$ matrix $\bG$ is a generator matrix of 
an $[n,k,d]_q$ error-correcting code if and only if
every $n-d+1$ columns of $\bG$ has rank $k$. 
\end{lemma} 
\vspace{-3pt}
\begin{lemma}
\label{lem:weak_security}
A coding scheme based on the matrix $\bG$ for an $(n,k)$-SMAN
is weakly secure if and only if every $k-1$ columns of $\bG$
generates an error-correcting code of minimum distance at least two. 
\end{lemma}
\vspace{-3pt}
\begin{proof}
This is a corollary of \cite[Lemma~3]{DauSongYuenISIT2014_2}. 
More details can be found in Appendix~\ref{app:a}. 
\end{proof}
\vspace{-3pt}
\begin{lemma}
\label{lem:distance2}
If the $\ell \times k$ matrix $\bA$ generates a $[k,\ell,d\geq 2]_q$ 
error-correcting code then 
\begin{equation}
\label{eq:ws_row}
|\cup_{j \in J} \supp(\bA_j)| \geq |J| + 1, \quad
\forall \varnothing \neq J \subseteq [\ell]. 
\end{equation}
\end{lemma}
\vspace{-3pt}
\begin{proof}
Suppose that $\bA$ generates a code of minimum distance at least two
but (\ref{eq:ws_row}) is violated.  
Then there exists $\varnothing \neq J \subseteq [\ell]$ such that \vspace{-3pt}
\begin{equation} 
\label{eq:1}
|\cup_{j \in J} \supp(\bA_j)| \leq |J|.  \vspace{-3pt}
\end{equation}
We aim to obtain a contradiction.

Let $I \subseteq [k] \setminus \cup_{j \in J} \supp(\bA_j)$
such that $|I| = k - |J|$. Moreover, let $L \subset [k]$
such that $L \supseteq I$ and $|L| = k-1$.
Let $\bA[L]$ be the $\ell \times (k-1)$ submatrix of $\bA$ that 
consists of columns of $\bA$ indexed by the elements in $L$.
Then according to Lemma~\ref{lem:GM_distance}, we have \vspace{-3pt}
\begin{equation} 
\label{eq:2}
\rank(\bA[L]) = \ell. \vspace{-3pt}
\end{equation}  
On the other hand, we claim that the $|J|$ rows of $\bA[L]$ indexed by
the elements in $J$ has rank at most $|J|-1$. 
As the remaining $\ell-|J|$ rows of $\bA[L]$ has rank at most $\ell-|J|$, 
we deduce that \vspace{-3pt}
\begin{equation} 
\label{eq:3}
\rank(\bA[L]) \leq (|J|-1) + (\ell-|J|) < \ell.
\end{equation} 
From (\ref{eq:2}) and (\ref{eq:3}) we obtain a contradiction.

We now prove that our aforementioned claim is correct. 
Consider the submatrix $\bA_J[L]$ that consists of 
rows  of $\bA[L]$ indexed by the elements of $J$. Due to (\ref{eq:1})
and our assumption that $L \supseteq I$, the submatrix 
$\bA_J[L]$ has at least $k - |J|$ all-zero columns. 
Since $|L| = k - 1$, $\bA_J[L]$ has $k-1$ columns.
Therefore, it has at most $|J|-1$ nonzero columns. 
Hence, $\rank(\bA_J[L]) \leq |J|-1$, as claimed.    
\end{proof}

\begin{remark}
\label{rm:1}
The result in Lemma~\ref{lem:distance2} can be extended
to $d \geq d'$ for any $d' \geq 1$ by replacing $|J|+1$ with $|J| + d'-1$ in (\ref{eq:ws_row}). 
\end{remark}

\begin{lemma}
\label{lem:var}
Let $\bP$ be a $(k-1)\times k$ $0$-$1$ matrix. 
Let $\vP$ be the matrix obtained from $\bP$ by replacing every nonzero entry of $\bP$ by some indeterminate $\xi_{i,j}$ over $\fq$. 
Suppose that all of these indeterminates are independent. 
Let $f(\vP) = \prod_{\bQ} \det(\bQ)$, where the product is taken over all $k$ submatrices $\bQ$ of order $k-1$ of $\vP$. 
Then $f(\vP)$, which is
a multivariable polynomial in $\fq[\cdots,\xi_{i,j},\cdots]$, is not identically
zero if and only if
\begin{equation} 
\label{eq:4}
|\cup_{j\in J} \supp(\bP_j)| \geq |J|+1,\quad
\forall \varnothing \neq J \subseteq [k-1]. 
\end{equation} 
\end{lemma}
\begin{proof}
The proof follows from \cite[Lemma 2-4]{DauSongDongYuenISIT2013}.
More details can be found in Appendix~\ref{app:b}.  
\end{proof}
\vskip 3pt

We are now in position to prove Theorem~\ref{thm:main1}. 
\begin{proof}[Proof of Theorem~\ref{thm:main1}]
\mbox{}\\
\textbf{Only-If.}
Suppose that there exists a weakly secure MDS coding scheme for 
an $(n,k)$-SMAN described by the adjacency matrix $\bM$. 
We aim to prove that the Weak Security Condition (\ref{eq:weak_security}) holds.
Let $\bG$ be the encoding matrix of the weakly secure MDS coding scheme. Note that as $\bG$ generates an MDS code, every subset of 
$k-1$ columns of $\bG$ is always linearly independent~\cite[Ch. 11]{MW_S}.  
Hence, by Lemma~\ref{lem:weak_security}, every set of $k-1$ columns of  
$\bG$ must generate a $[k,k-1,2]$ error-correcting code. 
Note here that $\supp(\bG[j]) \subseteq \supp(\bM[j])$ for all $j \in [n]$.
Hence, by applying Lemma~\ref{lem:distance2} to all $(k-1)\times k$
matrices corresponding to all subsets of $k-1$ columns of $\bG$, it 
is straightforward that the Weak Security Condition holds.

\nin\textbf{If.}
We assume that the Weak Security Condition holds, i.e.
\[
|\cup_{j \in J} \supp(\bM[j])| \geq |J|+1,\quad
\forall \varnothing \neq J \subset [n], \ |J| \leq k-1.
\] 
We aim to show that there exists a weakly secure MDS coding scheme
for the corresponding $(n,k)$-SMAN.

Using the same notation as in Lemma~\ref{lem:var}, let $\vM = (v_{i,j})$ where $v_{i,j} = 0$ if $m_{i,j} = 0$ and $v_{i,j} = \xi_{i,j}$
if $m_{i,j} \neq 0$. 
Here $\xi_{i,j}$'s are independent indeterminates.  
For each submatrix $\bP'$ of size $k \times (k-1)$ of $\bM$, let $\bP$
be its transpose and $\vP$ the corresponding (transposed) submatrix of $\vM$. We henceforth refer to such a matrix $\bP$ as a \emph{transposed submatrix} of $\bM$.  
Note that the Weak Security Condition (\ref{eq:weak_security}) on $\bM$  
implies the condition (\ref{eq:4}) on every transposed submatrix
$\bP$ of size $(k-1)\times k$ of $\bM$. Hence, by Lemma~\ref{lem:var}, the polynomial
$f(\vP)$ is not identically zero. Let
\[
F(\vM) = \prod_{\bP} f(\vP) \in \fq[\cdots \xi_{i,j} \cdots],
\]
where the product is taken over all transposed submatrices $\bP$ of size $(k-1) \times k$ of $\bM$. 
Then $F(\vM) \not\equiv \bO$. 

It is obvious that the Weak Security Condition (\ref{eq:weak_security}) implies the MDS Condition~\cite{DauSongYuenISIT14,DauSongYuen_JSAC2014, YanSprintsonZelenko2014}, which requires that every subset of $\ell$ relays must be collectively connected to at least $\ell$ sources, for all $\ell \leq k$.
Hence, if $f(\vM)$ is the product of determinants of all
submatrices of order $k$ of $\vM$ then $f(\vM) \not\equiv \bO$, 
according to~\cite[Lemma 2-4]{DauSongDongYuenISIT2013}. 
Therefore 
\[
F^{\text{ext}}(\vM) \define f(\vM)\times F(\vM) \not\equiv \bO.
\]
Hence, according to \cite[Lemma 4]{Ho2006},
for sufficiently large $q$, there exists 
$g_{i,j} \in \fq$ (for $(i,j)$ where $m_{i,j} = 1$) such that
\[
F^{\text{ext}}(\vM)(\cdots, g_{i,j}, \cdots) \neq 0.  \vspace{-3pt}
\]
As a consequence, 
\begin{equation} 
\label{eq:5}
f(\vP)(\cdots, g_{i,j}, \cdots) \neq 0,
\end{equation}  
for every transposed submatrix $\bP$ of size $(k-1)\times k$ 
of $\bM$. 
Let $\bG = (g_{i,j})$ (if $m_{i,j} = 0$ we set $g_{i,j} = 0$). 
Then thanks to (\ref{eq:5}), every transposed submatrix $\bA$ of size $(k-1)\times k$ of $\bG$ satisfies the following property: all submatrices of order $k-1$ of $\bA$ are invertible. 
Hence, according to \cite[Ch. 11]{MW_S}, every set of $k-1$ columns of $\bG$ generates an MDS $[k,k-1,2]_q$ error-correcting code. 
Thus, by Lemma~\ref{lem:weak_security}, the coding scheme based
on $\bG$ is weakly secure.    
Moreover, as $f(\vM)(\cdots, g_{i,j},\cdots) \neq 0$ as well, 
it follows that every submatrix of order $k$ of $\bG$ is invertible.
Thus, $\bG$ also generates an MDS code. 
\end{proof}
\vspace{-3pt}
\begin{remark}
Theorem~\ref{thm:main1} shows what the additional
cost is (in terms of source-relay links) when a passive adversary is also present, on top of an active adversary.
More specifically, while defending against an active adversary 
requires that every subset of $\ell$ relays must be collectively connected to at least $\ell$ sources, for all $\ell \leq k$, 
defending against both adversaries requires that every subset of $\ell$ relays must be collectively connected to at least $\ell+1$ sources, for all $0 < \ell < k$.
\end{remark} 

\subsection{Verification of Weak Security Condition in Polynomial Time}

While designing a weakly secure MDS coding scheme for a given $(n,k)$-SMAN
may require non-polynomial time (as random coding over finite fields with exponentially large sizes is used), verifying whether a SMAN supports a weakly
secure MDS coding scheme can be done in polynomial time. 
We prove this fact below using a proper modification of the proof of
\cite[Lemma 10]{DauSongYuen_JSAC2014}. 
We first present a simple lemma. Its proof is similar 
to the proof of~\cite[Lemma 4]{DauSongDongYuenISIT2013}
and can be found in Appendix~\ref{app:c}. 
\vspace{-3pt}
\begin{lemma}
\label{lem:ws_row}
The Weak Security Condition (\ref{eq:weak_security}) is equivalent to the following:
\begin{equation} 
\label{eq:ws_rowM}
|\cup_{i \in I} \supp(\bM_i)| \geq n - k + |I| + 1, \quad
\forall \varnothing \neq I \subsetneq [k]. 
\end{equation} 
\end{lemma}
\begin{proposition}
\label{pro:verify_ws}
The Weak Security Condition (\ref{eq:weak_security}) can be verified
in polynomial time in $n$ and $k$. 
\end{proposition}
\vspace{-3pt}
\begin{proof}
By Lemma~\ref{lem:ws_row}, it suffices to prove that (\ref{eq:ws_rowM}) can be
verified in polynomial time for all $\varnothing \neq I \subseteq
[k] \setminus \{i_0\}$, for every $i_0 \in [k]$.
Without loss of generality, let $i_0 = k$.  
The other cases can be proved in the same manner. 
We associate with $\bM$ a network $\nm$ constructed as follows. 
The set of nodes of $\nm$ consists of
\begin{itemize}
	\item a source node $s$,
	\item $n$ \emph{packet} nodes $s_1, \ldots, s_n$, 
	\item $k-1$ \emph{coding} nodes $r_1, \ldots, r_{k-1}$, 
	\item $k-1$ \emph{broadcast} nodes $b_1, \ldots, b_{k-1}$, 
	\item $k-1$ sink nodes $t_1, \ldots, t_{k-1}$. 
\end{itemize}
To simplify the notation, set $R_i = \supp(\bM_i)$. 
The set of directed edges of $\nm$ consists of
\begin{itemize}
	\item one edge of capacity \emph{one} from $s$ to $s_i$, $\forall i \in [n]$, 
	\item one edge of capacity \emph{infinity} from $s_j$ to $r_i$ if $j \in R_i$, 
	\item one edge of capacity \emph{one} from $r_i$ to $b_i$, $\forall i \in [k-1]$,
	\item one edge of capacity \emph{infinity} from $r_i$ to $t_i$, $\forall i \in [k-1]$, 
	\item one edge of capacity \emph{infinity} from $b_i$ to $t_j$, $\forall i, j \in [k-1]$. 
\end{itemize}
For instance, for the $(6,4)$-SMAN in Fig.~\ref{fig:SMAN-ex},
the corresponding adjacency matrix is
\begin{equation} 
\label{eq:mat}
\bM = 
\begin{pmatrix}
1 & 1 & 1 & 0 & 0 & 0\\
1 & 0 & 0 & 1 & 1 & 0\\
0 & 1 & 0 & 1 & 0 & 1\\
0 & 0 & 1 & 0 & 1 & 1
\end{pmatrix}, 
\end{equation} 
and the corresponding network $\mathcal{N}^4(\bM)$ is depicted in Fig.~\ref{fig:nm}. 
\begin{figure}[ht]
\centering
\includegraphics[scale=0.8]{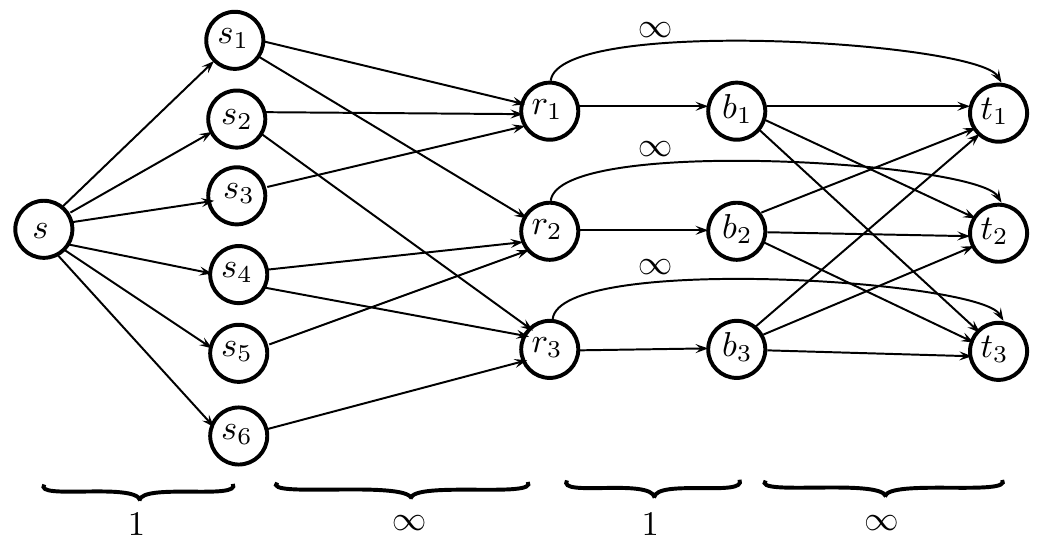}
\caption{The network $\mathcal{N}^4(\bM)$ associated to $\bM$ in (\ref{eq:mat})
when $i_0 = k = 4$.}
\label{fig:nm}
\vspace{-15pt}
\end{figure}

A cut $(S,T)$ of a network is a partition of the set of nodes of that network into 
two parts, namely $S$ and $T$. We are only interested in cuts that \emph{separate}
the source and some sink, i.e. $S$ contains the source and $T$ contains some sink. 
Let $c(u,v)$ denote the capacity of an edge $(u,v)$ in a network. Then the capacity 
of a cut $(S,T)$ is defined as
\[
c(S,T) = \sum_{u \in S, v\in T} c(u,v). 
\]
Consider the following \emph{Min-Cut} Condition for $\nm$: the capacity of every cut  
that separates $s$ and any sink is at least $n$. 
According to the Network Flow Algorithm (Ford-Fulkerson Algorithm), 
we can verify the Min-Cut Condition for $\nm$ in polynomial time.
Therefore, it suffices for our purpose to show that the condition 
(\ref{eq:ws_rowM}) restricted to those $I \subseteq [k-1]$ (for $\bM$)
is equivalent to the Min-Cut Condition (for $\nm$).  

Suppose that the Min-Cut Condition for $\nm$ holds. 
We aim to prove that (\ref{eq:ws_rowM}) restricted to those $I \subseteq [k-1]$ also holds for $\bM$. 
Let $I$ be an arbitrary nonempty subset of $[k-1]$. 
Recall that we use $R_i$ to denote $\supp(\bM_i)$.
Consider a cut $(S,T)$ where 
\[
\begin{split}
T = \{t_i\colon i \in I\} &\cup \{b_i\colon 1 \leq i \leq k-1\} \cup \{r_i\colon i \in I\}\\ 
&\cup \big\{s_j\colon j \in \cup_{i \in I} R_i \big\}.  
\end{split}
\]
Then the capacity of $(S,T)$ is 
\[
c(S,T) = \sum_{j \in \cup_{i \in I} R_i} c(s,s_j) + \sum_{i \notin I} c(r_i,b_i) = |\cup_{i \in I} R_i| + k - 1 - |I|.  
\]
As $c(S,T) \geq n$, we have 
\[
|\cup_{i \in I} R_i| \geq n - k + |I| + 1. 
\]
Conversely, suppose that (\ref{eq:ws_rowM}) restricted to those $I \subseteq [k-1]$ holds. 
We need to prove that $c(S,T) \geq n$ for every cut $(S,T)$ that separates $s$ and some sink. 
Suppose that $\{t_i\}_{i \in I'} \subseteq T$, where $\varnothing \neq I' \subseteq [k-1]$,
and that $t_i \notin T$ if $i \notin I'$. 
If $c(S,T) = \infty$ then it is larger than $n$ trivially. 
Now suppose that $c(S,T) < \infty$.  
Then $(S,T)$ does not contain any edge of the form $(s_j, r_i)$, $(r_i, t_i)$, or $(b_j, t_i)$, 
as these have capacity infinity.
Hence, $T$ must contain the following nodes 
\begin{itemize}
	\item $t_i$ for all $i \in I'$, because of our definition of $(S,T)$,
  \item $b_j$ for all $j \in [k-1]$, as $c(b_j,t_i) = \infty$ for every $j$ and $i$,
	\item $r_i$ for all $i \in I'$, as $c(r_i,t_i) = \infty$ for every $i$,
\end{itemize} 
Let $I$ be the subset of $[k-1]$ that satisfies  \vspace{-2pt}
\[
T \cap \{r_i\}_{i \in [k-1]} = \{r_i\}_{i \in I}. \vspace{-2pt}
\]
Then $I' \subseteq I$. Since $c(s_j,r_i) = \infty$ when $j \in R_i$, the set $T$ 
must also contains the packet nodes $s_j$ if $j \in R_i$ for some $i \in I$. 
Therefore, \vspace{-2pt}
\[
\begin{split} 
c(S,T) &\geq \sum_{j \in \cup_{i \in I} R_i} c(s,s_j) + \sum_{i \notin I} c(r_i,b_i)\\
&=|\cup_{i \in I} R_i| + k - 1 - |I| \\
&\geq (n - k + |I| + 1) + k - 1 - |I| = n.
\end{split} \vspace{-2pt}
\]
We complete the proof.
\end{proof}

\subsection{Trimming SMAN While Preserving Weak Security in Polynomial Time}

Given an $(n,k)$-SMAN that satisfies the Weak Security Condition 
(\ref{eq:weak_security}), 
Theorem~\ref{thm:trimming} states that one can trim the network to obtain
a sparsest possible network where the Weak Security Condition is 
still satisfied. Moreover, the trimming process can be done in polynomial time in $n$ and $k$. 
We prove this theorem in Appendix~\ref{app:d}. 
Note that we use here the equivalent statement of 
the Weak Security Condition stated in Lemma~\ref{lem:ws_row}.

\begin{theorem}
\label{thm:trimming}
For each $i \in [k]$ let $R_i$ be an arbitrary subset of $[n]$ $(n \geq k)$. Suppose that \vspace{-2pt}
\begin{equation}
\label{eq:7}
|\cup_{i \in I} R_i | \geq n - k + |I| + 1,\quad \forall
\varnothing \neq I \subsetneq [k].  \vspace{-2pt}
\end{equation}
Then for every $i \in [k]$ there exists
a subset $R'_i \subseteq R_i$ such that
\begin{itemize}
	\item $|\cup_{i \in I} R'_i | \geq n - k + |I| + 1$,
	$\forall \varnothing \neq I \subsetneq [k]$,
	\item $|R'_i| = n - k + 2$, for all $i \in [k]$. 
\end{itemize}
Moreover, such subsets $R'_i$ can be found in polynomial time.
\end{theorem}

\section{Extension to block security}
\label{sec:block_security}

In this section we extend our result on weak security to a more general concept of \emph{block security} (or \emph{security against guessing} in some other works in the network coding literature). 

The coding scheme for an $(n,k)$-SMAN based on a $k \times n$ encoding matrix $\bG$ is \emph{$b_\ell$-block secure} against a passive 
adversary of \emph{strength} $\ell$ $(\ell < k)$ if the conditional entropy \vspace{-2pt}
\[
\entropy\big(\{X_j: j \in B\} \mid \{\bX \bG[j] : j \in E\}\big) = \entropy(\{X_j: j \in B\}), \vspace{-2pt}
\] 
for every subset $B \subset [k]$, $|B| \leq b_\ell$, and for every subset $E \subset [n]$, $|E| \leq \ell$.
In words, a coding scheme is $b_\ell$-block secure against an 
adversary of strength $\ell$ if an adversary that eavesdrops on
at most $\ell$ relay-sink packets obtains no information about
each subset of at most $b_\ell$ source packets. 
In that case, even if the adversary can guess correctly some
$b_\ell-1$ source packets, it still gains no information about any other packet. 

\begin{theorem}
\label{thm:main2}
An $(n,k)$-SMAN supports an MDS coding scheme that is $b_\ell$-block
secure against an adversary of strength $\ell$ if and only if
\begin{equation} 
\label{eq:block}
|\cup_{j \in J}\supp(\bM[j])| \geq |J| + b_\ell,\
\forall \varnothing \neq J \subset [n],\ |J| \leq \ell,
\end{equation} 
where $\bM[j]$ denotes the $j$th column of the adjacency matrix
$\bM$. Note that in the right-hand side of (\ref{eq:block}), 
the first term $|J|$ corresponds to the MDS Condition, while the second term $b_\ell$ corresponds to the block security level. 
\end{theorem}

Theorem~\ref{thm:main2} characterizes the block security level of 
an MDS coding scheme for SMAN based on the density of the source-relay
links. In a special case where the SMAN is densest, i.e. each source is connected to all relays, then according to~\cite{OliveiraLimaVinhozaBarrosMedard2012}, a Cauchy matrix would provide 
the best level of block security, which is $b_\ell = k - \ell$.  
The proof of Theorem~\ref{thm:main2} follows the same idea of
that of Theorem~\ref{thm:main1}, using a generalized
version of Lemma~\ref{lem:distance2} (see Remark~\ref{rm:1}).   
We omit the proof. 

While the Weak Security Condition can be verified in polynomial time, 
it is not known whether a similar conclusion holds for block security.
More specifically, given an $(n,k)$-SMAN and a sequence $\{b_\ell\}_1^{k-1}$, 
whether (\ref{eq:block}) can be verified in polynomial time is 
still an open question. 
 
\section*{Acknowledgment}
This work is funded by iTrust. 
\bibliographystyle{IEEEtran}
\bibliography{WeaklySecureECC}

\appendix

\subsection{Proof of Lemma~\ref{lem:weak_security}}
\label{app:a}

This lemma is a corollary of \cite[Lemma~3]{DauSongYuenISIT2014_2}.
We refer the reader to \cite{DauSongYuenISIT2014_2} and its extended
version for a detailed and rigorous proof. 
Nevertheless, we provide here an informal proof that may illustrate better the intuition behind this lemma.
Below we refer to the number of nonzero coordinates of a vector
as its (Hamming) \emph{weight}. It is well known in coding theory 
that the minimum Hamming weight of a nonzero codeword of any linear error-correcting code is equal to its minimum distance.   

Suppose that every $k-1$ columns of the coding matrix $\bG$ generates an error correcting codes of minimum distance $d \geq 2$. 
Suppose that a passive adversary obtains $\bx \bG[E]$, where $\bG[E]$ is the $k \times (k-1)$ submatrix of $\bG$ formed by columns indexed
by $E$, for some $E \subset [n]$, $|E|=k-1$.  
Hence, it can linearly transform these $k-1$ coded symbols
by considering the product 
\[
\bx\bG[E]\boldsymbol{\alpha}^\text{t} = \bx (\boldsymbol{\alpha}\bG[E]^\text{t})^\text{t}
=\bx \bc^\text{t},
\]
where $\boldsymbol{\alpha} \in \fq^{k-1}$ is some coefficient vector, the superscript ``t" denotes the transpose, and $\bc = \boldsymbol{\alpha}\bG[E]^\text{t}$. 
Since the columns of $\bG[E]$ generates an error-correcting codes of minimum distance at least two, $\bc$, if nonzero, has weight at least two. In other words, if $\bc \neq \bO$ then it 
has at least two nonzero coordinates. 
As a result, $\bx \bc^\text{t}$ is a linear combination of at least two source packets. 
Therefore, by linearly transforming the eavesdropped coded symbols 
$\bx\bG[E]$, the adversary cannot determine explicitly each source packet. As the source packets are independent and uniformly 
randomly distributed over $\fq$, this is equivalent to saying that
the conditional entropy of each source packet remains the same
given the knowledge of $k-1$ coded packets.
Hence, the coding scheme is weakly secure.  

Conversely, if for some subset $E \subset [n]$, $|E| = k-1$, the columns of $\bG[E]$ generate a linear error-correcting code of minimum distance one, then there exists $\boldsymbol{\alpha}
\in \fq^{k-1}$ such that $\bc = \boldsymbol{\alpha}\bG[E]^\text{t}$
has weight one. Suppose that $c_i \neq 0$ and $c_j = 0$ if $j \neq i$. 
Then by post-multiplying $\bx\bG[E]$ by 
$\boldsymbol{\alpha}^\text{t}$, the adversary obtains the source packet
$x_i$ explicitly. Hence, in this case the coding scheme is not
weakly secure.  

\subsection{Proof of Lemma~\ref{lem:var}}
\label{app:b} 

This lemma is a corollary of \cite[Lemma 2-4]{DauSongDongYuenISIT2013}. Indeed, let $\bM$ be a $k \times n$ binary matrix. Then 
\cite[Lemma 2-4]{DauSongDongYuenISIT2013} conclude that
$f(\vM) \not\equiv \bO$ if and only if 
\[
|\cup_{i \in I} \supp(\bM_i)| \geq n - k + |I|, \quad 
\forall \varnothing \neq I \subseteq [k].
\] 
Applying this conclusion to the $(k-1) \times k$ matrix $\bP$
in Lemma~\ref{lem:var}, the proof follows. 

\subsection{Proof of Lemma~\ref{lem:ws_row}}
\label{app:c}

Suppose that (\ref{eq:weak_security}) does not hold, 
i.e. there exists $\varnothing \neq J \subset [n]$, 
$|J| \leq k-1$, such that
\begin{equation}
\label{eq:6}
|\cup_{j \in J} \supp(\bM[j])| \leq |J|.
\end{equation}
We aim to show that (\ref{eq:ws_rowM}) does not hold either. 
Indeed, from (\ref{eq:6}), let $I \subseteq [k] \setminus 
\cup_{j \in J} \supp(\bM[j])$ such that $|I| = k - |J|$. 
Because $1 \leq |J| \leq k-1$, we deduce that $\varnothing \neq I
\subsetneq [k]$. Moreover, due to (\ref{eq:6}) and our assumption that
$I \subseteq [k] \setminus \cup_{j \in J} \supp(\bM[j])$, we conclude
that
\[
|\cup_{i \in I} \supp(\bM_i)| \leq n - |J| = n - k + |I|. 
\] 
Hence, (\ref{eq:ws_rowM}) is violated. 

Conversely, we need to show that if (\ref{eq:ws_rowM}) does not hold then 
neither does (\ref{eq:weak_security}). The proof is completely 
similar and therefore is omitted. 

\subsection{Proof of Theorem~\ref{thm:trimming}}
\label{app:d}
We can prove this theorem by modifying the proof of
\cite[Theorem 2]{DauSongYuenISIT14} accordingly.
Both proofs follow the same idea of a well-known proof of Hall's marriage theorem: repeatedly removing the edges of the bipartite graph
until the graph becomes sparsest yet still satisfies the Hall's condition. 
To simplify the notation, for a set $I \subseteq [k]$ we use $R_I$ to denote the union $\cup_{i \in I}R_i$. 

Suppose that the sets $R_i$ satisfy (\ref{eq:7}). We keep removing the elements of these
sets while maintaining the Weak Security Condition (\ref{eq:7}). Assume that at some point, the removal
of any element in any set $R_i$ would make them violate (\ref{eq:7}). 
We prove that now the sets $R_i$ have cardinality precisely $n-k+2$, which concludes the
first part of the theorem. 

Suppose, for contradiction, that there exists $r \in [k]$ such that $|R_r| \geq n - k + 3$. 
Take $a$ and $b$ in $R_r$, $a \neq b$. For all $i \in [k]$, let 
\begin{equation} 
\label{eq:old2}
R^a_i = 
\begin{cases}
R_i \setminus \{a\},& \text{ if } i = r,\\
R_i,& \text{ otherwise,} 
\end{cases}
\end{equation} 
\begin{equation} 
\label{eq:old3}
R^b_i = 
\begin{cases}
R_i \setminus \{b\},& \text{ if } i = r,\\
R_i,& \text{ otherwise.} 
\end{cases}
\end{equation} 
According to our assumption, both of the two collections of sets $\{R^a_i\}_{i \in [k]}$
and $\{R^b_i\}_{i \in [k]}$ violate (\ref{eq:7}). 
Therefore, there exist two nonempty subsets $A \subseteq [k]$ and $B \subseteq [k]$, 
$r \notin A \cup B$, such that 
\begin{equation} 
\label{eq:old4}
|R^a_{A \cup \{r\}}| < n - k + |A| + 2,
\end{equation} 
\begin{equation} 
\label{eq:old5}
|R^b_{B \cup \{r\}}| < n - k + |B| + 2. 
\end{equation} 
Since $r \notin A$, by (\ref{eq:old2}) we have
\begin{equation}
\label{eq:old6}
|R^a_{A \cup \{r\}}| \geq |R^a_A| = |R_A| \geq n - k + |A| + 1.  
\end{equation} 
Similarly, since $r \notin B$, by (\ref{eq:old3}) we have
\begin{equation}
\label{eq:old7}
|R^b_{B \cup \{r\}}| \geq |R^b_B| = |R_B| \geq n - k + |B| + 1.  
\end{equation} 
From (\ref{eq:old4}) and (\ref{eq:old6}) we deduce that
\begin{equation}
\label{eq:old8}
|R^a_{A \cup \{r\}}| = |R^a_A| = |R_A| = n - k + |A| + 1.
\end{equation} 
Similarly, from (\ref{eq:old5}) and (\ref{eq:old7}) we have
\begin{equation}
\label{eq:old9}
|R^b_{B \cup \{r\}}| = |R^b_B| = |R_B| = n - k + |B| + 1.  
\end{equation} 
Therefore, 
\begin{equation} 
\label{eq:old10}
R^a_{A \cup \{r\}} \cap R^b_{B \cup \{r\}} = R_A \cap R_B.
\end{equation} 
Moreover, as $a \in R^b_{B \cup \{r\}}$ and $b \in R^a_{A \cup \{r\}}$, we deduce that 
\begin{equation} 
\label{eq:old11}
R^a_{A \cup \{r\}} \cup R^b_{B \cup \{r\}} = R_{A \cup B \cup \{r\}}. 
\end{equation} 
From (\ref{eq:old8}) and (\ref{eq:old9}) we have
\begin{equation} 
\label{eq:old13}
\begin{split}
&\quad\ 2(n-k) + |A| + |B| + 2\\
&= |R^a_{A \cup \{r\}}| + |R^b_{B \cup \{r\}}|\\
&= |R^a_{A \cup \{r\}} \cup R^b_{B \cup \{r\}}| + |R^a_{A \cup \{r\}} \cap R^b_{B \cup \{r\}}|\\ 
&= |R_{A \cup B \cup \{r\}}| + |R_A \cap R_B|,
\end{split} 
\end{equation}  
where the last transition is due to (\ref{eq:old10}) and (\ref{eq:old11}). 
We further evaluate the two terms of the last sum in (\ref{eq:old13}) as follows. 
The first term
\begin{equation} 
\label{eq:old14}
\begin{split} 
|R_{A \cup B \cup \{r\}}| &\geq n - k + |A \cup B \cup \{r\}| + 1\\
&= n - k + |A \cup B| + 2. 
\end{split} 
\end{equation} 
The second term
\begin{equation} 
\label{eq:old15}
|R_A \cap R_B| \geq n - k + |A \cap B| + 1,
\end{equation} 
which can be explained below. 
\begin{itemize}
  \item If $A \cap B \neq \varnothing$, then by applying (\ref{eq:7}) to $A \cap B$ we obtain
	\[
	|R_A \cap R_B| \geq |R_{A \cap B}| \geq n - k + |A\cap B| + 1. 
	\]
	\item If $A \cap B = \varnothing$, then $n - k + |A \cap B| + 1 = n - k + 1$.
We have
\begin{equation} 
\label{eq:old16}
R^a_{A \cup \{r\}} = R^a_A \cup R^a_r = R_A \cup (R_r \setminus \{a\}). 		
\end{equation} 
By (\ref{eq:old8}), $R^a_{A \cup \{r\}} = R_A$. Combining this with (\ref{eq:old16}) we deduce that
\begin{equation} 
\label{eq:old17}
R_r \setminus \{a\} \subseteq R_A. 
\end{equation}  
Similarly, 
\begin{equation} 
\label{eq:old18}
R_r \setminus \{b\} \subseteq R_B. 
\end{equation}  
From (\ref{eq:old17}) and (\ref{eq:old18}) we have
	\[
	|R_A \cap R_B| \geq |R_r \setminus \{a,b\}| \geq (n - k + 3) - 2 = n - k + 1,   
	\]
which proves that (\ref{eq:old15}) is correct when $A \cap B = \varnothing$. 
\end{itemize}
Finally, from (\ref{eq:old13}), (\ref{eq:old14}), and (\ref{eq:old15}) we deduce that
\[
\begin{split} 
&\quad \ 2(n-k) + |A| + |B| + 2\\
&\geq \big(n - k + |A \cup B| + 2\big) + \big(n - k + |A \cap B| + 1 \big)\\
&= 2(n-k) + |A| + |B| + 3,
\end{split} 
\]
which produces a contradiction. 

The proof of the first part of this theorem also provides a polynomial time
algorithm to find subsets of $R_i$'s that all have cardinality $n -k+2$ yet still
maintain the Weak Security Condition (\ref{eq:7}). 
Indeed, we keep removing the
elements of the subsets $R_i$ in the following way. If there exists $r \in [k]$
such that $|R_r| \geq n - k + 3$, then as we just prove, for $a, b \in R_r$, 
it is impossible that removing $a$ or $b$ from $R_r$ both render the Weak Security Condition
violated. Therefore, we can either remove $a$ or $b$ while still maintaining
the Weak Security Condition. Note that by Proposition~\ref{pro:verify_ws}, 
the Weak Security Condition can be verified in polynomial time in $k$ and $n$. 
Therefore, this algorithm terminates in polynomial time in $k$ and $n$ 
and produces subsets $R'_i$'s of the original sets $R_i$'s that satisfy the
stated requirement in the theorem.  
Note that by setting $I = \{i\}$, the Weak Security Condition (\ref{eq:7}) implies
that $|R'_i| \geq n - k + 2$. Hence, those $R'_i$'s form a sparsest
$(n,k)$-SMAN that still supports a weakly secure MDS coding scheme. 
\end{document}